\documentclass[11pt,reqno]{amsart}

\usepackage{amssymb,amsmath,amsthm,graphics,epsfig}
\usepackage{amsfonts}
\usepackage{graphicx,cite}
\usepackage{amssymb}
\usepackage[pagewise]{lineno}
\usepackage{mathrsfs}
\usepackage[colorlinks=true, pdfstartview=FitV, linkcolor=blue, citecolor=blue, urlcolor=blue]{hyperref}
\usepackage{booktabs}
\usepackage{graphicx}
\usepackage{tikz}
\usepackage{subfigure}
\usepackage{subcaption}
\usepackage{multicol}
\usepackage{cases}
\usepackage{booktabs}
\usepackage{hyperref}
\usepackage{arydshln}
\usepackage[colorinlistoftodos]{todonotes}
\usepackage[normalem]{ulem}
\usepackage{color}
\usepackage{amsaddr} 
\usepackage[margin=2.7cm]{geometry}
\usepackage{lineno}

\newtheorem{Lem}{Lemma}
\newtheorem{The}{Theorem}
\newtheorem{Pro}{Proposition}

\usepackage{float}
\theoremstyle{definition}

\newtheorem{Rem}{Remark}

\definecolor{ColorEdward}{rgb}{0.4,0.6,0.7}

\definecolor{ColorFrancisco}{rgb}{0.956,0.878,0.5}

\definecolor{ColorJosep}{rgb}{0.2,0.5,0.1}

\newcommand{\sgn}{\operatorname{sgn}}
\DeclareMathOperator{\R}{\mathbb{R}}

\begin{document}

\title{Quasispecies dynamics with time lags and periodic fluctuations in replication}

\author{Edward A. Turner$^{1,2,*}$, Francisco crespo$^{2}$, Josep Sardanyés$^{3}$\\ \MakeLowercase{and} Nolbert Morales$^4$}

\address{$^{1}${Universidad de Concepción, Departamento de Ciencias Básicas, Campus Los Ángeles,\\ {\small \itshape{Av. J.A. Coloma 0201, Los Ángeles, Chile}}}}
	
\address{$^{2}${Universidad del Bío-Bío, Grupo de investigaci\'on en Sistemas Din\'amicos\\ y Aplicaciones (GISDA), {\small \itshape{Casilla 5-C, Concepci\'on, Chile}}}}
\address{$^{3}$Centre de Recerca Matem\`atica (CRM), 
{\small \itshape{08193 Cerdanyola del Vall\`es,\\ Barcelona, Catalonia, Spain }}}
\address{$^{4}$Universidad San Sebastián, Facultad de Ingenier\'ia, Arquitectura y Diseño,\\ {\small \itshape{Lago Panguipulli 1390, Puerto Montt, Chile }}}

\thanks{$^{*}$Corresponding author: \texttt{edward.turner1901@alumnos.ubiobio.cl}}

	

\begin{abstract}
		Quasispecies theory provides the conceptual and theoretical bases for describing the dynamics of biological information of replicators subject to large mutation rates. This theory, initially conceived within the framework of prebiotic evolution, is also being used to investigate the evolutionary dynamics of RNA viruses and heterogeneous cancer cells populations. In this sense, efforts to approximate the initial quasispecies theory to more realistic scenarios have been made in recent decades. Despite this, how time lags in RNA synthesis and periodic fluctuations impact quasispecies dynamics remains poorly studied. In this article, we combine the theory of delayed ordinary differential equations and topological Leray-Schauder degree to investigate the classical quasispecies model in the single-peak fitness landscape considering time lags and periodic fluctuations in replication. First, we prove that the dynamics with time lags under the constant population constraint remains in the simplex in both forward and backward times. With backward mutation and periodic fluctuations, we prove the existence of periodic orbits regardless of time lags. Nevertheless, without backward mutation, neither periodic fluctuation nor the introduction of time lags leads to periodic orbits. However, in the case of periodic fluctuations, solutions converge exponentially to a periodic oscillation around the equilibria associated with a constant replication rate. We check the validity of the error catastrophe hypothesis assuming no backward mutation; we determine that the error threshold remains sound for the case of time of periodic fitness and time lags with constant fitness. Finally, our results show that the error threshold is not found with backward mutations.
\end{abstract}
\maketitle	


\section{Introduction}
Quasispecies theory was conceived in the 1970's by Manfred Eigen~\cite{Eigen1971} and Peter Schuster~\cite{Eigen1988}. This theory was developed to investigate the dynamics of biological information for replicators subject to large mutation rates and was initially applied within the framework of prebiotic evolution. Later on, this theory was adapted to systems of replicons evolving under large mutation rates such as RNA viruses~\cite{Mas2004,Perales2020,Revull2021} and cells with remarkable genetic instability in cancer~\cite{Sole2003,Sole2004,Brumera2006}.  More recently, a conceptual parallelism between viral quasispecies and the conformational heterogeneity of prions have been established~\cite{Li2010,Weissmann2011,Weissmann2012}. The investigations into experimental quasispecies involving bacterial, animal, and plant RNA viruses, along with its implications for RNA genetics, were early summarized in Refs.~\cite{Domingo1985,Domingo1988}. 

The concept quasispecies refers to highly heterogeneous populations of replicators composed by the so-called master or wild-type (wt) sequence which is surrounded by a cloud of mutants (mutant spectrum) which stabilize at the mutation-selection balance~\cite{Eigen1971,Eigen1988,Nowak1991,Bull2005}. Hence, selection acts on the quasispecies as a whole more than on a particular sequence or set of sequences. 
The integration of quasispecies theory into virology has fundamentally transformed our comprehension of the composition and dynamics of viral populations during disease onset. The presence of a mutant spectrum was initially demonstrated through clonal analyses of RNA bacteriophage Q$\beta$ populations in an infection initated with a single viral particle~\cite{Domingo1978}. Since this finding, viral quasispecies have been identified and quantified in multitude of viruses such as foot-and-mouth disease virus~\cite{Domingo1980,Sobrino1983}, vesicular stomatitis virus~\cite{Holland1979,Holland1982}, hepatitis viruses~\cite{Martell1992,Davis1999,Mas2004,Perales2020}, or SARS-CoV-2~\cite{Domingo2023}, to cite some examples. 

One of the most ground-breaking predictions of quasispecies theory is the so-called error threshold or error catastrophe~\cite{Eigen1971,Eigen1988,Biebricher2005}. This notion has also led to other important concepts such as lethal mutagenesis and lethal defection. In the classic quasispecies model, the error threshold is the mutation rate beyond which the master sequences i.e., the sequence with the highest replicative capacity, fades out and the population is exclusively composed of mutant sequences~\cite{Eigen1971,Biebricher2005}. As a difference from lethal mutagenesis (see below), which involves an effective extinction of sequences, the error threshld involves a shift in the sequence space due to the surpass of the critical mutation involving a full population composed by mutants. It is well known that, under the single-peak fitness landscape assumption, the error threshold is governed by a transcritical bifurcation~\cite{Sole2004, Sole2006,Castillo2017}. Increased mutation rates typically result in decreased population fitness as most mutations with phenotypic effects are detrimental. This principle underlies the concept of lethal mutagenesis, where a viral population can be eradicated by intentionally inducing mutations through mutagenic agents~\cite{Bull2008}. Evidence of lethal mutagenesis have been provided i.a., for human immunodeficiency virus type 1~\cite{Loeb1999,Dapp2013}, poliovirus~\cite{Crotty2001}, food-and-mouth-disease virus~\cite{Perales2011,Avila2017}, and hepatitis C virus~\cite{Prieto2013,Avila2016} in cell cultures (see also~\cite{Perales2019}). Lethal defection, which involves the extinction of viral populations due to appearance of defective viral genomes at low amounts of mutagen, has been identified for lymphocytic choriomeningitis virus un cell cultures~\cite{Grande2005}. 

Initial models of viral quasispecies inherited assumptions of quasispecies theory developed for prebiotic replicators. These included, for instance, deterministic dynamics (infinite populations) and geometric replication~\cite{Eigen1971,Eigen1988}, and simple fitness landscapes such as the Swetina-Schuster landscape~\cite{Swetina1982,Sole2003,Sole2004,Sole2006}. This simple fitness landscape considers two different populations given by the master sequence and the pool of mutants which are averaged over a single variable. Despite the simplicity of this approach, a simple model with such assumptions qualitatively explained complexity features tied to the error threshold in hepatitis C-infected patients~\cite{Sole2006,Mas2004}. 

However, RNA viruses unfold an enormous complexity at the genetic and population dynamics levels~\cite{Sole2021,Sanjuan2021,Aylward2022}. During the last decades, considerable efforts have been made to approach the initial quasispecies theory to more realistic scenarios for RNA viruses and for quasispecies theory in general. These new models have considered, either separately or in combination: finite populations~\cite{Nowak1989,Sole2004,Sardanyes2008}; stochastic effects~\cite{Sole2004,Sardanyes2011,Ari2016}; spatially-embedded quasispecies~\cite{Altemeyer2001,Sardanyes2008,Sardanyes2011}; more complex fitness landscapes including dynamic ones~\cite{Wilke2001a,Wilke2001c}, epistasis~\cite{Sardanyes2009,Elena2010}, and mutational fitness effects~\cite{Josep2014}; viral complementation~\cite{Sardanyes2010}; the survival-of-the-flattest effect (with empirical evidence in the evolution of computer programs~\cite{Wilke2001} and viroids~\cite{Codoner2006}), see also~\cite{Wilke2001b,Sardanyes2008}; and asymmetric modes of replication~\cite{Sardanyes2009,Sardanyes2011}, among others.

Despite these efforts, there are still some features of viral (and RNA) dynamics that have not been considered within quasispecies theory. For instance, the impact of time lags in the replication of RNAs. Quasispecies theory, and most of the models for RNA virus replication, consider instantaneous synthesis of full genomes (either wildtype or distinct mutants), thus ignoring the time needed for the synthesis of the genomes. However, viral genomes are synthesized by the sequential incorporation of nucleotides by the RNA-dependent RNA-polymerase in the replication complex. Several studies have quantified the rates of elongation of viral genomes. For instance, recent studies have revealed that the SARS-CoV-2 replication complex has an elongation rate of $150$ to $200$ nucleotides per second~\cite{Campagnola2022}. This elongation rate is more than twice as fast the poliovirus polymerase complex. According to these rates, a full SARS-CoV-2 RNA genome ($\approx 29.9$ Kb) is synthesized in about $2.5$ and $3.32$ minutes. A full poliovirus genome ($\approx 7.4$ Kb) will be synthesized in about $1.5$  minutes.

Another, often ignored, important effect are the fluctuations that key parameters such as replication rates may suffer. For RNA viruses, specially those infecting plants, replication processes may follow periodic fluctuations at the within-tissue or within-host levels at different time scales mainly due to to changes in temperature~\cite{Obreepalska2015,Honjo2020}. Moreover, it is known that the mammalian brain has an endogenous central circadian clock that regulates central and peripheral cellular activities. At the molecular level, this day-night cycle induces the expression of upstream and downstream transcription factors that influence the immune system and the severity of viral infections over time. In addition, there are also circadian effects on host tolerance pathways. This stimulates adaptation to normal changes in environmental conditions and requirements (including light and food). These rhythms influence the pharmacokinetics and efficacy of therapeutic drugs and vaccines. The importance of circadian systems in regulating viral infections and the host response to viruses is currently of great importance for clinical management~\cite{Zandi2023}. 

In this article we aim to covering this gap in quasispecies models by studying the impact of time lags and fluctuations in RNA replication. To do so we use the Swetina-Schuster fitness landdscape~\cite{Swetina1982,Sole2006} as a first approach to this problem. The article is organised as follows. In Section~\ref{sec:general_model} we introduce the classical quasispecies model and prove that the constant population assumption is valid to study time lags in dynamics. Section~\ref{sec:landscape}  introduces the model under the single-peak fitness landscape assumption. Here, we first investigate the dynamics without backward mutation considering periodic replicative fitness and time lags in replication separately. Finally, we also consider the case where phenotypic reversions occur considering periodic repliation rates. 

\section{General quasispecies model with time lags and periodic fitness} \label{sec:general_model}
The classical quasispecies model was initially formulated within the framework of well-stirred populations of replicators assuming a constant population (CP)~\cite{Eigen1971,Eigen1988}. This system can be modelled by the following system of autonomous ordinary differential equations: 
\begin{equation}\label{eq:Quasi}
	\dot x_{i}=\sum_{j=0}^n f_jQ_{ji}x_j - \tilde\Phi(\mathbf{x})x_{i}.
\end{equation}
The state variables $x_i$ denote the concentration (population numbers) of the $i$-th replicator species. Parameter $f_j$ is the replication rate of the $j$-th population of replicators, $Q_{ji}$ is the matrix denoting the transitions from the $j$-th to the $i$-the population of replicators due to mutation, and $\tilde\Phi(\mathbf{x})=\sum_{j=0}^nf_jx_j$ is the out-flux term ensuring a CP i.e., $\sum_i \dot{x}_i = 0$ and $\sum_i x_i = constant$. Note that the CP setting $\sum_i x_i = 1$ involves that the out-flux term is given by the average fitness associated to the population vector $\mathbf{x}=(x_0,\dots,x_n).$  The CP condition for Eqs.~\eqref{eq:Quasi} makes the orbits to span the simplex 
\begin{equation}\label{eq:CP}
	\Sigma_{n+1}=\left\lbrace \mathbf x\in\mathbb R^{n+1}:x_0+x_1+\dots+x_n=1,x_0,\dots,x_n>0 \right\rbrace.
\end{equation}
This feature is proved in the following proposition.

\begin{Pro}\label{Prop:1} The simplex \eqref{eq:CP} is invariant under the flow of system \eqref{eq:Quasi}.  
\end{Pro}
\begin{proof} Let observe that $\frac{d}{dt}\Sigma_{n+1}=D\Sigma_{n+1}(\mathbf x)\dot{\mathbf x}=\sum_{i=0}^{n}\dot x_{i}.$ Then, by adding the equations of \eqref{eq:Quasi} and considering that $\sum_{j=0}^n Q_{ji}=1$ for $i=0,\dots,n,$ we obtain that
	\begin{align*}
		\sum_{i=0}^n\dot x_i&=\tilde\Phi(\mathbf x)\left(1-\sum_{i=0}^nx_i\right)=0.
	\end{align*}
	Thus, $\Sigma_{n+1}$ is invariant.
\end{proof}

As we highlighted in the Introduction, the quasispecies model assumes that the replication of the populations is instantaneous. However, many biological processes suffer time lags. For instance the replication of RNAs or the proliferation of cancer cells do not occur instantaneously. If we focus on RNA viruses, viral RNA genomes are synthesized by the RNA-dependent RNA polymerase, which take some time in synthesizing the full RNA sequences. This polymerization and other phenomena such as RNA folding and maturation introduce time lags in the production of a mature, replicating RNA sequence. Moreover, viral replication can also change due to circadian cycles or day-night temperature fluctuations. Incorporating these effects into the quasispecies model is fundamental for a more comprehensive understanding and accurate representation of complex biological phenomena. The quasispecies model to explore these features in a qualitative manner is given by the next non-autonomous delay differential equation
\begin{equation}\label{eq:QuasiDelay}
	\dot x_{i}(t)=\sum_{j=0}^n f_j(t)Q_{ji}x_j(t-r_j)-x_{i}(t)\Phi(\mathbf x), 
\end{equation}
where now $f_j(t)$ is a periodic function and $r_j$ introduce the time lags (see below), and the outflow term reads $\Phi(\mathbf x)=\sum_{j=0}^n f_j(t)\,x_j(t-r_j)$.


As we mentioned above, the CP condition largely determines the dynamics of the quasispecies model which is tied to the simplex space~\eqref{eq:CP}. Next, we will show that this condition still remains valid after with time lags.


It is well-known that time lags can introduce major changes into a system of differential equations. Usually, key features of the system are modified and  many techniques employed to analyze ordinary differential systems are no longer in use. However, the CP constraint is compatible with the delayed model \eqref{eq:QuasiDelay}. Hence, Proposition \ref{Prop:1} can be extended to system \eqref{eq:QuasiDelay}. For this purpose, we proceed in two stages; first, Proposition~\ref{vari_inva} shows that initial conditions with right endpoints in $\Sigma_{n+1}$ will remain for all positive time in $\Sigma_{n+1}$. In other words, it is not necessary for the entire initial condition to reside completely within $\Sigma_{n+1}$. Later, Proposition~\ref{Perm_esp_inv} indicates that any completely positive $T$-periodic solution of the system must necessarily reside in $\Sigma_{n+1}$.

\begin{Pro}\label{vari_inva} 
	Let us consider $\phi_i\in C(\mathbb R)$ and $\gamma=\max_{j=0,n}\{r_{j}\}$. If $\phi_i(t_0)\in \Sigma_{n+1}$ with $t_0\in\R$, then $\Sigma_{n+1}$ is an  invariant space  of  delay differential equations system  
	\begin{equation}
	\begin{split}\label{Sist_Gen}
		&\dot x_{i}(t)=\sum_{j=0}^n f_j(t)Q_{ji}x_j(t-r_j)-x_{i}(t)\Phi(\mathbf x),\\
		&x_i(t)=\phi_i(t),\quad t\in\left[t_0-\gamma,t_0\right],  \quad i=0,\dots,n.
	\end{split}
	\end{equation}
\end{Pro}
\begin{proof} Let $x_i(t)\in C\left([t_0,\infty[,\R^{n+1}\right)$ is a solution of the system \eqref{Sist_Gen}. By the Proposition \ref{Prop:1}, 
	$$\frac{d}{dt}\left(\displaystyle\sum_{i=0}^nx_i(t) \right)=\left(1-\displaystyle\sum_{i=0}^nx_i(t) \right)\displaystyle\sum_{j=0}^nf_j(t)x_j(t-r_{j}).$$
	If we consider $z(t)=\sum_{i=0}^nx_i(t)$ and considering the method of steps in system \eqref{Sist_Gen}, this system reduces for $t\in\left[0, \gamma\right]$ in the ordinary differential equation 
	\begin{equation}\label{EDO}
		\dot z(t) =(1-z(t))\displaystyle\sum_{j=0}^nf_j(t) \phi_j(t-r_{j}), 
	\end{equation}
	such that $z(t_0)=1$, then by the existence and uniqueness theorem we have that $z(t)=1$ for  $t\in [t_0,t_0+\gamma]$,  applying again the method of the steps for $t\in[t_0+\gamma,t_0+2\gamma]$ we obtain the equation \eqref{EDO} with the initial condition $z(t_0+\gamma)=1$  from where again we have that  $z(t)=1$ for  $t\in [ t_0+\gamma,t_0+2\gamma]$. If we continue in this way by induction we obtain that $z(t)=\sum_{i=0}^nx_i(t)=1$  for all $t\geq t_0$.
\end{proof}

\begin{Pro}\label{Perm_esp_inv} 
	If $\phi_i\in C([0,T],\R)$ are positive $T$-periodical solutions of the system \eqref{vari_inva} in $]0,1[^{n+1}$, then $(\phi_0(t),...,\phi_n(t))\in \Sigma_{n+1}$ for all $t\in\R$.  
\end{Pro}
\begin{proof}
	Let us assume that $\phi_i \in C([0,T],\mathbb{R})$  with $i=0,...,n$ are positive $T$-periodic solutions of the system \eqref{vari_inva}  in $]0,1[^{n+1}$, not contained in $\Sigma_{n+1}$. Then, the function $z(t) = \sum_{i=0}^n \phi_i(t)$ is also a positive periodic function that satisfies the differential equation
	\begin{equation}\nonumber
		\dot z(t) =(1-z(t))\displaystyle\sum_{j=0}^nf_j(t)\phi_j(t-r_{j}(t)).
	\end{equation}
	Since $z$ is a $T$-periodic function, there exists $t_0\in[0,T]$ such that 
	\begin{equation}\nonumber
		0 =(1-z(t_0))\displaystyle\sum_{j=0}^nf_j(t_0)\phi_j(t-r_{j}(t_0)),
	\end{equation}
	and therefore $(\phi_0(t_0),...,\phi_n(t_0))\in \Sigma_{n+1}$. Then, if we consider the initial value problem
	\begin{equation}\nonumber
		\dot x_i(t)=\sum_{j=0}^nf_j(t)x_j(t-r_{j}(t))(Q_{ji}-x_i(t)), \quad x_i(t)=\phi_i(t),\quad t\in\left[t_0-T,t_0\right].
	\end{equation}
	From the Proposition \ref{vari_inva} we obtain that $(x_0(t),...,x_n(t)) \in \Sigma_{n+1}$ for all $ t\geq t_0.$  Then, due to the uniqueness of solutions in delayed differential equations, we can conclude that  $(x_0(t),...,x_n(t))=(\phi_0(t),...,\phi_n(t))\in \Sigma_{n+1}$ for all $t\in \R$.
\end{proof}

These statements provide our system with a biological interpretation, with and without time lags. 

\section{Quasispecies model for the single-peak fitness landscape}
\label{sec:landscape}
In this section we introduce the quasispecies model in a simple fitness landscape which assumes that mutations generate deleterious mutants. This allows to divide the population of sequences into two state variables: the master sequence $x_0$ with a high fitness, and the pool of mutants which are lumped together into a, lower fitness, average sequence $x_1$. This fitness landscape, also known as the Swetina-Schuster landscape~\cite{Swetina1982,Sole2003} is explored here for two reasons: (i) it provides a suitable mathematical framework allowing for analytical derivations, specially under the CP for which the dynamical system can be reduced by a degree of freedom; (ii) such a landscape recovered quasispecies complexity features in clinical data~\cite{Sole2006}. This system is obtained by considering $n=1$ in \eqref{eq:QuasiDelay} 
 \begin{equation}\label{eq}
	\begin{split}
		\dot x_{0}&=f_{0}(1-\mu)x_{0}(t-r_0)+f_1\xi x_1(t-r_1)-x_{0}(t)\left(f_0x_0(t-r_0)+f_1x_1(t-r_1)\right),\\
		\dot x_{1}&=f_{0}\mu x_{0}(t-r_0)+f_{1}(1-\xi)x_{1}(t-r_1)-x_{1}(t)\left(f_0x_0(t-r_0)+f_1x_1(t-r_1)\right),
	\end{split}
\end{equation}
where $Q_{01}=\mu,$ $Q_{00}=1-\mu,$ $Q_{11}=(1-\xi),$  and $Q_{10}=\xi.$ In the subsequent sections, we also employ the following compact notation for \eqref{eq}
$$\dot{\mathbf{x}}(t) = F(\mathbf x(t),\mathbf x(t-r_0),\mathbf x(t-r_1)),\quad \mathbf x=(x_0,x_1),$$ 
which is a nonlinear retarded functional differential system with two bounded lags, and the map $F:\mathbb R^6\to\mathbb R$. As mentioned, $f_i$ denotes the replication rates, and $\mu$ and $\xi$ are the mutation rates, both forward and backward. For $n=1$, the dynamics span the segment $\Sigma_{2}$, hereafter denoted as $\Sigma$ for simplicity. 

It is important to note that our analyses will distinguish the absence and presence of backward mutation, which correspond with $\xi = 0$ and $\xi \neq 0$, respectively. Some works have studied the single-peak fitness landscapes without backward mutations~\cite{Sole2003,Sole2004,Sole2006}. This approach assumes that mutations occur from the master to the mutant populations but not in the reverse sense. The enormous size of the sequence space makes this assumption a good first approximation. However, models with backward mutations have been also explored~\cite{Sardanyes2009,Sardanyes2011,Josep2014}. Despite backward point-mutations producing the original master sequences may be very improbable, beneficial mutations causing phenotypic reversions producing sequences with the same fitness than the master sequence may occur. Moreover, our investigation primarily explores how periodic fitness and time lags affect the qualitative behavior of the flow, specifically through the identification of oscillatory phenomena for both cases.

\subsection{Dynamics with no backward mutation}
The single-peak landscape model without backward mutation is obtained by considering $\xi=0$ in \eqref{eq}, obtaining
\begin{equation}\label{eq:16}
	\begin{split}
		\dot x_{0}&=f_{0}(t)(1-\mu)x_{0}(t-r_0)-x_{0}(t)\Phi(\mathbf x) \\
		\dot x_{1}&=f_{0}(t)\mu x_{0}(t-r_0)+f_{1}(t)x_{1}(t-r_1)-x_{1}(t)\Phi(\mathbf x) .
	\end{split}
\end{equation}
where $f_j$ represents positive $T$-periodic continuous functions, and $\mu \in \left[0,1\right]$ denotes the mutation rate of $x_{0}$, here with $\Phi(\mathbf x) =f_0(t)x_0(t-r_0)+f_1(t)x_1(t-r_1)$.
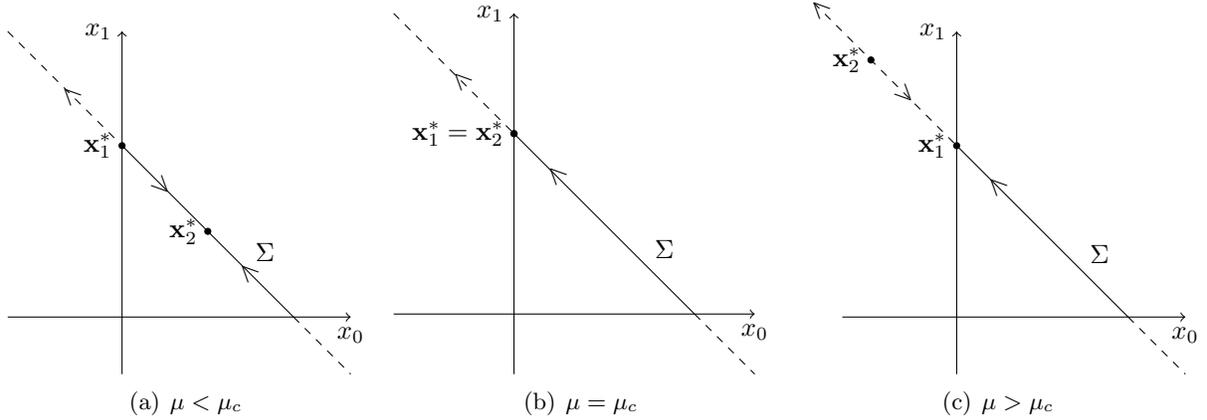
\begin{figure}
	\begin{center}
	\subfigure[$\mu<\mu_c$]{
		\begin{tikzpicture}[scale=0.76]
			\draw[->] (-2,0) -- (4,0);
			\draw[->] (0,-1) -- (0,5);
			\draw [] (0,3) -- (3,0);
			\draw[dashed] (-2,5) -- (0,3);
			\draw[dashed] (3,0) -- (4,-1);
			\filldraw[color=black] (0,3) circle (1.5pt);
			\filldraw[color=black] (1.5,1.5) circle (1.5pt);
			\coordinate [label=above:\textcolor{black}{\small$\Sigma$}] (E2) at (2.5,0.8);
			\coordinate [label=below:\textcolor{black}{\small$x_0$}] (E2) at (4,0);
			\coordinate [label=left:\textcolor{black}{\small$x_1$}] (E2) at (0,5);
			\coordinate [label=left:\textcolor{black}{\small$\mathbf x_1^*$}] (E2) at (0,3);
			\coordinate [label=left:\textcolor{black}{\small$\mathbf x_2^*$}] (E2) at (1.5,1.5);
			\coordinate [label=left:\textcolor{black}{\rotatebox{-45}{\small$<$}}] (E2) at (-0.5,3.9);
			\coordinate [label=below:\textcolor{black}{\rotatebox{135}{\small$<$}}] (E2) at (0.7,2.7);
			\coordinate [label=below:\textcolor{black}{\rotatebox{-45}{\small$<$}}] (E2) at (2.2,1.2);
		\end{tikzpicture}}
		\subfigure[$\mu=\mu_c$]{
		\begin{tikzpicture}[scale=0.8]
			\draw[->] (-2,0) -- (4,0);
			\draw[->] (0,-1) -- (0,5);
			\draw [] (0,3) -- (3,0);
			\draw[dashed] (-2,5) -- (0,3);
			\draw[dashed] (3,0) -- (4,-1);
			\filldraw[color=black] (0,3) circle (1.5pt);
			\coordinate [label=above:\textcolor{black}{\small$\Sigma$}] (E2) at (2.5,0.76);
			\coordinate [label=below:\textcolor{black}{\small$x_0$}] (E2) at (4,0);
			\coordinate [label=left:\textcolor{black}{\small$x_1$}] (E2) at (0,5);
			\coordinate [label=left:\textcolor{black}{\small$\mathbf x_1^*=\mathbf x_2^*$}] (E2) at (0,3);
			\coordinate [label=below:\textcolor{black}{\rotatebox{-45}{\small$<$}}] (E2) at (0.7,2.7);
			\coordinate [label=left:\textcolor{black}{\rotatebox{-45}{\small$<$}}] (E2) at (-0.5,3.9);
		\end{tikzpicture}}
		\subfigure[$\mu>\mu_c$]{
		\begin{tikzpicture}[scale=0.76]
			\draw[->] (-2,0) -- (4,0);
			\draw[->] (0,-1) -- (0,5);
			\draw [] (0,3) -- (3,0);
			\draw[dashed] (-2.5,5.5) -- (0,3);
			\draw[dashed] (3,0) -- (4,-1);
			\filldraw[color=black] (0,3) circle (1.5pt);
			\filldraw[color=black] (-1.5,4.5) circle (1.5pt);
			\coordinate [label=above:\textcolor{black}{\small$\Sigma$}] (E2) at (2.5,0.77);
			\coordinate [label=below:\textcolor{black}{\small$x_0$}] (E2) at (4,0);
			\coordinate [label=left:\textcolor{black}{\small$x_1$}] (E2) at (0,5);
			\coordinate [label=left:\textcolor{black}{\small$\mathbf x_1^*$}] (E2) at (0,3);
			\coordinate [label=left:\textcolor{black}{\small$\mathbf x_2^*$}] (E2) at (-1.5,4.5);
			\coordinate [label=below:\textcolor{black}{\rotatebox{-45}{\small$<$}}] (E2) at (0.7,2.7);
			\coordinate [label=left:\textcolor{black}{\rotatebox{135}{\small$<$}}] (E2) at (-0.5,3.9);
			\coordinate [label=below:\textcolor{black}{\rotatebox{-45}{\small$<$}}] (E2) at (-2.4,5.8);
		\end{tikzpicture}}
	\end{center}
 \captionsetup{width=\linewidth}
 \caption{Transcritical bifurcation associated with the error threshold in the quasispecies model. The segment $\Sigma$ is the biological meaningful region, which is contained in the attraction basin of the equilibria $\mathbf x_1^*$ for $\mu\geq\mu_c$.}
	\label{fig:Trans}
\end{figure}

\begin{Rem}
\label{Remark:Erro}
The classical cuasispecies model corresponds with $r_i=0$ and constant fitness $f_i(t)=f_i$. This system has three equilibrium points; $\mathbf x_0^*=(0,0)\notin \Sigma,$ $\mathbf x_1^*=(0,1)$ and 
\begin{equation}\label{eq:equilibioEsp}
\mathbf x_2^*=(x_2^*,y_2^*)=\left(\frac{{f_0}-{f_1}-{f_0}\mu}{{f_0}-{f_1}},\frac{{f_0}\mu}{f_0-{f_1}}\right),
\end{equation}
where $\mathbf x_1^*,\mathbf x_2^*\in\Sigma$. A transcritical bifurcation occurs at the critical mutation rate $\mu_c=1-f_1/f_0$, indicating a threshold beyond which selection becomes impossible, leading the system into a drift phase~\cite{Sole2006,Sole2021}. The equilibrium $\mathbf x_1^*$ is unstable and $\mathbf x_2^*$ is globally asymptotically stable for $\mu < \mu_c$. At $\mu = \mu_c$ both equilibria $\mathbf x_1^*$ and $\mathbf x_2^*$ collide and interchange stability in a transcritical bifurcation. Hence, for $\mu > \mu_c$, $\mathbf x_1^*$ is globally asymptotically stable and  $\mathbf x_2^*$ unstable (Fig.~\ref{fig:Trans}).

The error threshold is significantly relevant within the quasispecies model. However, when $\xi\neq 0$, this threshold is never reached and $\mathbf x_1^*$ is not an equilibrium point of the system, as discussed in \cite{CrespoCancer2020}.
\end{Rem}

\subsubsection{Periodic replication rates with no time lags} 

This section analyzes the dynamics considering periodic fluctuations in the replication rate. One might expect that a periodic replicative fitness would lead to periodic behavior of the solutions. However, our study will prove otherwise. 

If we consider $f_i(t)=m_i+n\cos(t)$, the CP condition $x_1=1-x_0$ and a reparametrization of the independent variable, we can express system \eqref{eq:16} in terms of the master sequence $x_0$, denoted by $x$ for simplicity. The model reads
\begin{equation}\label{eq:sistred}
	\dot x=(1+a\cos(t))(1-\mu)x-x((1+a\cos(t))x+(r+a\cos(t))(1-x)),
\end{equation}
where $a=n/m_0$ and $r=m_1/m_0.$ This equation has only one equilibria $x^*=0$, whose stability changes with $\mu$. Precisely, following Remark~\ref{Remark:Erro}, we know that $x^*$ is unstable for $\mu\neq\mu_c=1-r$. However, inspecting the dynamics in $\Sigma$, the equilibrium $x^*$ is stable for $\mu\geq\mu_c$. Moreover, considering $x(0)=c$ the solution of \eqref{eq:sistred} can be explicitly computed and is given by
\begin{eqnarray}
\label{eq:solucion}  x(t)&=& \dfrac{ce^{(\mu_c-\mu)t-a\mu\sin(t)}}{1+c\mu_c\displaystyle\int_0^t e^{(\mu_c-\mu)s-a\mu\sin(s)}ds}.
\end{eqnarray}
Note that for any initial condition and parameters, the numerator of \eqref{eq:solucion} is periodic, while the denominator is not. Therefore, the system does not posses periodic solutions for any value of the parameters. We will distinguish several regimes for $x(t)$ according to the sign of $\mu_c-\mu.$ 

\begin{figure}
\begin{tikzpicture}[scale=1]
		\node[inner sep=0pt] (caso1) at (-4.3,0) {\includegraphics[scale=0.8]{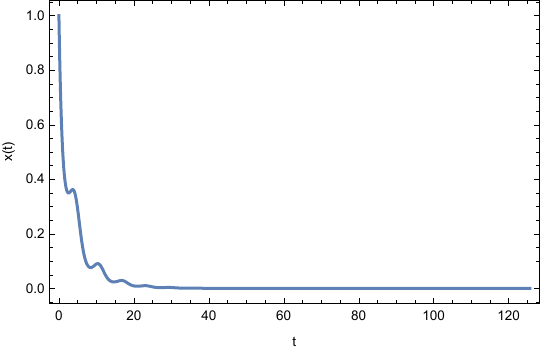}};
		\node[inner sep=0pt] (caso1) at (4.3,0) {\includegraphics[scale=0.8]{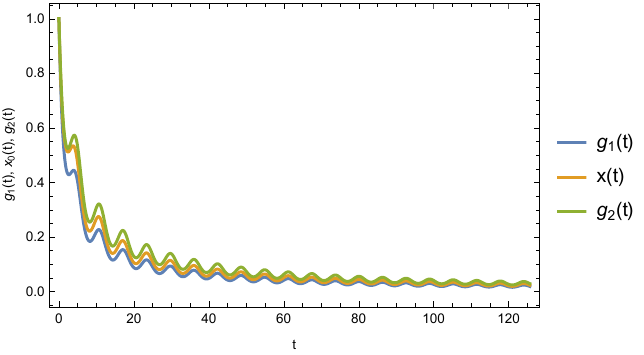}};
			\coordinate [label=below:\textcolor{black}{\small$\mu>\mu_c$}] (E2) at (-4,2);
			\coordinate [label=below:\textcolor{black}{\small$\mu=\mu_c$}] (E2) at (4,2);
		\end{tikzpicture}
   \captionsetup{width=\linewidth}
	\caption{Drift phase arising beyond the critical mutation rate $\mu_c$, with the dominance of mutant sequences i.e., $x(t)=0$ for $m_0=0.3,$ $m_1=0.2,$ $n=0.2,$ $\mu=0.1$.}
	\label{fig:3}
\end{figure}

If $\mu > \mu_c,$ previous treatment of the non-periodic case \cite{Eigen1971,Schuster1994} predicts a lost of genomic information as the population enters into a drift phase. This fact is preserved with a periodic replicative fitness. In one hand, if $\mu=\mu_c$ we can bound the solution \eqref{eq:solucion} as follows
\begin{equation}
  g_1(t)=\dfrac{c}{e^{a\mu \sin(t)}(1+c\mu_c e^{a\mu} t)} \leq x(t)\leq \dfrac{c}{e^{a\mu \sin(t)}(1+c\mu_c e^{-a\mu} t)}=g_2(t).
\end{equation}
Since $g_i(t)\to 0$ as $t\to\infty$ for $i=1,2$, we have that $x(t)\to0$. On the other hand, considering $\mu>\mu_c$, we have that the denominator in \eqref{eq:solucion} is a positive number greater than 1, and the numerator tends to zero as $t\to\infty$. Thus, we conclude that the $x(t)\to0$. We illustrate this behavior in Figure~\ref{fig:3}. 

In the case where $\mu<\mu_c,$ the solution is bounded by 
\begin{equation}
    h_1(t)=\frac{e^{-a \mu  \sin (t)}}{k_{11}+k_{12} e^{-t(\mu_c-\mu)}}\leq x(t) \leq \frac{e^{-a \mu  \sin (t)}}{k_{21}+k_{22} e^{-t(\mu_c-\mu)}}=h_2(t),
\end{equation}
where
\begin{equation}
 k_{11}=\frac{\mu_c e^{a \mu }}{\mu_c-\mu },\quad k_{12}=\frac{\mu_c e^{a \mu }}{\mu -\mu_c}+\frac{1}{c}\quad\text{and}\quad k_{21}=\frac{\mu_c e^{-a \mu }}{\mu_c-\mu },\quad k_{22}=\frac{\mu_c e^{-a \mu }}{\mu -\mu_c}+\frac{1}{c}.
\end{equation}
To understand the behavior of the bounds $h_1(t)$ and $h_2(t)$, we observe that these function are not periodic, but they converge exponentially to the following periodic functions
\begin{equation}
    p_i(t)=\frac{e^{-a \mu  \sin (t)}}{k_{i1}},\quad i=1,2.
\end{equation}
Moreover, the functions $p_i(t)$ are related to the equilibrium $\mathbf x_2^*$ of the corresponding constant fitness model given in \eqref{eq:equilibioEsp}. In particular, we have 
\begin{equation}
\label{eq:Promedios}
\dfrac{1}{2\pi}\int_0^{2\pi} p_1(t) \,dt=k_{11}^{-1}=e^{-a \mu }\, x_2^*,\qquad \dfrac{1}{2\pi}\int_0^{2\pi} p_2(t) \,dt=k_{21}^{-1}=e^{a \mu } \,x_2^*.
\end{equation}
As a result, after a initial drifting stage, the solution $x(t)$ shows a quasi-periodic behavior and is bounded by two quasi-periodic functions oscillating below and above the equilibria of the constant fitness model. We illustrate this feature in Figure~\ref{fig:2}, which is a representative example of numerical simulations for $x(t)$.

It is noteworthy that the bounds $h_i(t)$ converge to $p_i(t)$ independently of the initial condition $x(0)=c$. Therefore, the oscillation around the average values given in \eqref{eq:Promedios} play the role of an attractor for arbitrary solutions. In practice, every solution exhibits a first stage moving toward the bounds imposed by $p_i(t)$. After that, they behave almost periodically.
\begin{figure}
\begin{tikzpicture}[scale=1]
		\node[inner sep=0pt] (caso1) at (0,0) {\includegraphics[scale=1]{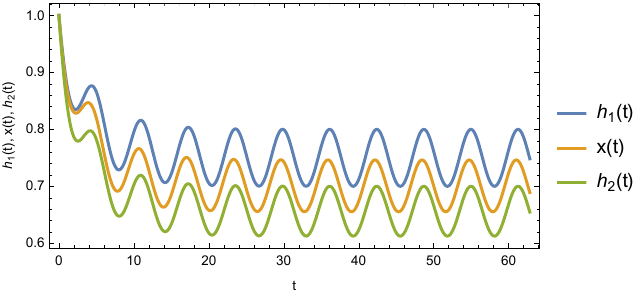}};
			\coordinate [label=below:\textcolor{black}{\small$\mu<\mu_c$}] (E2) at (-0.5,2.3);
		\end{tikzpicture}
	\caption{Quasi-periodic behavior of the the solutions of system \eqref{eq:sistred} for $m_0=0.3,$ $m_1=0.2,$ $n=0.2.$}
	\label{fig:2}
\end{figure}

\subsubsection{Dynamics with time lags and constant replication}
This section analyzes the model with time lags in the synthesis of the sequences assuming a constant replicative fitness. Let us consider system \eqref{eq:16} with $f_0=1,$ $f_1<1$ and $r_0=r_1.$ To simplify the notation, we will henceforth denote $f_1=f$ and $r_i=r$, taking the form
\begin{equation}\label{eq:Delayconstante}
	\begin{split}
		\dot x_{0}&=(1-\mu)x_{0}(t-r)-x_{0}(t)\left(x_0(t-r)+f x_1(t-r)\right) \\
		\dot x_{1}&=\mu x_{0}(t-r)+fx_{1}(t-r)-x_{1}(t)\left(x_0(t-r)+f x_1(t-r)\right).
	\end{split}
\end{equation}
This system has three equilibrium points: $\mathbf x_0^*=(0,0)\notin\Sigma,$ $\mathbf x_1^*=(x_1^*,y_1^*)=(0,1)$ and
\begin{equation}
 \mathbf x_2^*=(x_2^*,y_2^*)=\left(\frac{\mu+f-1}{f-1},\frac{-\mu}{f-1}\right),
\end{equation}
where $\mathbf x_1^*,\mathbf x_2^*\in\Sigma.$ The transcritical bifurcation mentioned in Remark~\ref{Remark:Erro} remains in the system \eqref{eq:Delayconstante} for the critical value $\mu_c=1-f,$ given us the conditions on which every equilibrium exists. 
Previously we showed that the simplex $\Sigma_n$ associated with the CP constraint is invariant even in the presence of lags. In particular, $\Sigma_n$ for system \eqref{eq:Delayconstante} is given by $$x_0(t)+x_1(t)=1.$$ 
Considering only the biological meaningful domine, we are left with the following reduced system
\begin{equation}\label{eq:1616}
\dot{x}(t) = G(x(t), x(t-r))=-f x(t)+(1-\mu)x(t-r)+(f-1)x(t)x(t-r),
\end{equation}
where the master sequence $x_0$ is denoted by $x.$ The reduced system is endowed with the following equilibrium points $x_1^*=0$ and $x_2^*=(\mu+f-1)/(f-1).$

Absolute stability is defined in \cite{Ruan2001}. This type of stability is the main obstacle for using the method given in \cite{Ruan2001} determining periodic orbits. Following the mentioned method, we linearize the equation around $x^*_1,$ $x_2^*$ and observe absolute stability for the system \eqref{eq:1616}.

\begin{The}
Let $x_1^*$ and $x_2^*$ be the equilibria of the equation \eqref{eq:1616}. Then,
\begin{enumerate}
\item[\textit{i)}] If $1-\mu<f,$ $x_1^*$ is unstable and $x_2^*$ is absolutely stable,
\item[\textit{ii)}] If $1-\mu>f,$ $x_1^*$ is absolutely stable and $x_2^*$ is unstable.
\end{enumerate}
\end{The}
\begin{proof}
In one hand, the linearized equation around $x_1^*=0$ is $\dot x(t)=(1-\mu)x(t-r)-f x(t).$
If we consider $x=e^{zt},$ the characteristic equation becomes
\begin{equation}\label{12}
z+f=(1-\mu)e^{-rz},
\end{equation}
which has infinite complex solutions. However, by the Theorem 4.7 of \cite{smith2010introduction} when $1-\mu<f,$ $x_1^*$ is unstable and $x_2^*$ is absolutely stable.
On the other hand,  the system around $x_2^*$ takes the form
\begin{equation}
\begin{split}
\frac{\partial G}{\partial x(t)}(x^*,x^*)&=-f+(f-1)x^*=\mu-1\\
\frac{\partial G}{\partial x(t-r)}(x^*,x^*)&=1-\mu+(f-1)x^*=f.
\end{split}
\end{equation}
The system becomes $\dot x(t)=(\mu-1)x(t)+fx(t-r)$ and the associated characteristic equation is
\begin{equation}\label{123}
z+(1-\mu)=f e^{-zr}.
\end{equation}
In the same way as the previous equilibrium, given the form of \eqref{123} and by the Theorem 4.7 of \cite{smith2010introduction} when  $1-\mu>f,$ $x_1^*$ is absolutely stable and $x_2^*$ is unstable.
\end{proof}

\begin{Rem}
The previous theorem establishes conditions for instability and absolute stability. The latest prevents the existence of periodic orbits around the equilibrium at hand, while this matter remains inconclusive for an unstable equilibrium. A classical mechanism to determine periodic orbits existence is the Hopf bifurcation. However, this procedure does not apply in our model.
\end{Rem}

\subsection{Dynamics with backward mutation and periodic replication rates}
\label{sec:4}

In this section, we consider arbitrary periodic replication rates $f_0$ and $f_1$ in the presence of backward mutation. For this setting, we establish the existence of at least one positive  $T$-periodic solution for the system \eqref{eq}. Our approach relies on topological Leray-Schauder degree arguments \cite{L_Schauder} and is either valid when considering time lags or not. In what follows, we consider $\mu,\xi\in]0,1[$ and define the ratio $w:=f_{0}/f_{1}$. 

If we first restrict $f_0$ and $f_1$ to be constant functions (period equal zero), it is easy to see that the system \eqref{eq} has three trivial periodic solutions, which are as follows:
\begin{itemize}
	\item If $w \neq 1$, considering $\alpha_1=1+\xi-w (1-\mu)$ and $\alpha_2=1-\xi+w (1-\mu)-2w$,
	\begin{equation}
		\begin{split}\label{equilibria}
			\mathbf{x}^*_0 &= (0,0) \\
			\mathbf{x}^*_1 &= \left(\dfrac{\alpha_1-\sqrt{\alpha_1^2-4\xi(1-w)}}{2(1-w)},\dfrac{\alpha_2+\sqrt{\alpha_1^2-4\xi(1-w)}}{2(1-w)}\right)\\
   			\mathbf{x}^*_2 &= \left(\dfrac{\alpha_1+\sqrt{\alpha_1^2-4\xi(1-w)}}{2(1-w)},\dfrac{\alpha_2-\sqrt{\alpha_1^2-4\xi(1-w)}}{2(1-w)}\right).
		\end{split}
	\end{equation}
	\item If $w=1$ (this case corresponds to neutral mutants),
	\begin{equation}
		\mathbf{x}^*_0 = (0,0), \quad \mathbf{x}^*_1 = \mathbf{x}^*_2 = \left(\dfrac{\xi}{\xi+\mu},\dfrac{\mu}{\xi+\mu}\right).
	\end{equation}
\end{itemize}

A basic analysis shows that $x_1$ and $x_2$ switch positions in the first quadrant as $r$ increases from $r<1$ to $r>1$. This fact ensures that at least one trivial solution always lies in the first quadrant. Moreover, we observe that by perturbing the values of $f_0$ and $f_1$ with non-constant periodic functions, the trivial periodic solutions mentioned above are no longer valid. Next, we will prove that these trivial periodic solutions located in the first quadrant become non-trivial periodic solution for periodic replication rates $f_0$ and $f_1$.



\begin{The}\label{TeoP}  
	Let us consider $\xi\neq 1-\mu$. Then, system \eqref{eq} admits at least one non-trivial $T$-periodic positive solution in $\Sigma$  for all $ t\in[0,T]$, such that,
	\begin{equation}\nonumber
		\min\{1-\mu,\xi\}\leq x_0(t)\leq \max\{1-\mu,\xi\},\quad \min\{1-\xi,\mu\}\leq x_1(t)\leq \max\{1-\xi,\mu\}.
	\end{equation}
\end{The}

\begin{The}\label{Teo2}  
	Let $x_0,x_1$ be positive T-periodic solutions of the system \eqref{eq} in the interval $]0,1[$. Then it is fulfilled that $(x_0(t),x_1(t))\in\Sigma$, $\forall t\in[0,T]$ and
	\begin{itemize}
		\item[\textit{i)}] If $1-\mu>\xi$, then 
		\begin{equation}\nonumber
			\xi<x_0(t)<(1-\mu),\ \forall t\in[0,T]\quad \text{and}\quad \mu<x_1(t)<1-\xi,\ \forall t\in[0,T].
		\end{equation}
		\item[\textit{ii)}] If $1-\mu<\xi$, then 
		\begin{equation}\nonumber
			(1-\mu)<x_0(t)<\xi,\ \forall t\in[0,T]\quad \text{and}\quad 
			1-\xi<x_1(t)<\mu,\ \forall t\in[0,T].
		\end{equation}
		\item[\textit{iii)}] If $(1-\mu)=\xi$, then $x_0(t)=\xi$ and $x_1(t)= 1-\xi$ for all $t\in[0,T].$	
	\end{itemize}
\end{The}

We prove these theorems in the following section.

\begin{Rem}
Notice that the previous results not only establish the existence of periodic orbits in the meaningful biological region $\Sigma$, but they also determine the bounds for the oscillations of these periodic orbits.
\end{Rem}

\begin{figure}
	\begin{center}
		\begin{tikzpicture}[scale=1.5]
			\draw[->] (0,0) -- (3.5,0);
			\draw[->] (0,0) -- (0,3.5);
			\draw [line width = 1pt] (0,3) -- (3,0);
			\draw [line width = 1pt, color = red] (0.75,2.25) -- (2.25,0.75);
			\draw[dashed] (0,2.25) -- (2.25,2.25);
			\draw[dashed] (0.75,0) -- (0.75,2.25);
			\draw[dashed] (0,0.75) -- (2.25,0.75);
			\draw[dashed] (2.25,0) -- (2.25,2.25);
			\coordinate [label=below:\textcolor{black}{\small$\Sigma$}] (E2) at (0.5,3.25);
			\coordinate [label=below:\textcolor{black}{\small$x_0$}] (E2) at (3.5,0);
			\coordinate [label=left:\textcolor{black}{\small$x_1$}] (E2) at (0,3.5);
			\coordinate [label=below:\textcolor{black}{\small $\nu$}] (E2) at (2.25,0);
			\coordinate [label=below:\textcolor{black}{\small$\eta$}] (E2) at (0.75,0);
			\coordinate [label=left:\textcolor{black}{\small$1-\eta$}] (E2) at (0,2.25);
			\coordinate [label=left:\textcolor{black}{\small$1-\nu$}] (E2) at (0,0.75);
			\fill[gray!35,nearly transparent] (0.75,0.75) -- (0.75,2.25) -- (2.25,2.25) -- (2.25,0.75) -- cycle;
		\end{tikzpicture}
	\end{center}
	\caption{In red, we can observe the segment containing the periodic orbit. 
	}
	\label{fig:PeriodicOrbit}
\end{figure}
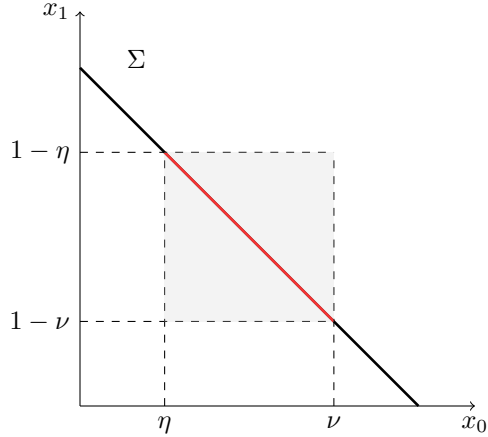

Firstly, we develop the theoretical framework to prove Theorem~\ref{TeoP} and Theorem~\ref{Teo2}. Then, the mentioned proofs are given at the end of this section.

We aim to reformulate the task of discovering $T$-periodic solutions by framing it as a fixed-point problem associated with a relatively compact operator for the system given in equation \eqref{eq}.

Let us define $C_T=\left\{ u \in C(\R): u(t-r)=u(t+r) \mbox{ for }t\in \R\right\}$ and the operator $$\mathcal{A} :C_T\times C_T\to C_T\times C_T,$$  
\begin{equation}
	\quad \mathcal{A} [x_0,x_1] =\mathcal P[x_0,x_1]+\mathcal Q\mathcal N[x_0,x_1] +\mathcal K(I-\mathcal Q)\mathcal N[x_0,x_1] ,
\end{equation}
where  $\mathcal P[x_0,x_1]=(x_0(0),x_1(0)),$ $\mathcal Q[x_0,x_1]=(\overline{x_0},\overline{x_1}),$ $\overline{x_0}=\frac{1}{T}\int_0^Tx_0(t)dt,$
$$\mathcal N[x_0,x_1] (t)=\left(\begin{array}{lr} f_{0}(t)(1-\mu)x_{0}(t-r_0)+f_1(t)\xi x_1(t-r_1)-x_{0}(t)\left(f_0(t)x_0(t-r_0)+f_1(t)x_1(t-r_1)\right)\\   f_{0}(t)\mu x_{0}(t-r_0)+f_{1}(t)(1-\xi)x_{1}(t-r_1)-x_{1}(t)\left(f_0(t)x_0(t-r_0)+f_1(t)x_1(t-r_1)\right)\end{array}\right)$$
and
$$\mathcal K[x_0,x_1]=\left(\int_0^t x_0(s)ds,\int_0^t x_1(s)ds\right).$$

Note that this operator is entirely continuous, and thus the \textit{Leray-Schauder Degree} is applicable. Furthermore, the periodic boundary value problem for \eqref{eq} is equivalently transformed into the fixed-point problem for an operator equation:
\begin{equation*}
	(x_0,x_1) =\mathcal{A}[x_0,x_1] ,\quad (x_0,x_1) \in C_T\times C_T.
\end{equation*}
The crucial step in establishing the validity of Theorems \ref{TeoP} is to demonstrate that the degree of the operator $I-\mathcal{A}$ on a suitable open set is non-zero.   Let $\eta, \nu \in \mathbb{R}$, $\eta =\min\{1 - \mu, \xi\}$, $\nu=\max\{1 - \mu, \xi\}  $ and consider $\xi \neq 1 - \mu$. Then, we define the open set
\begin{gather*}
	\Omega:=\left\{(x_0,x_1)\in C_T\times C_T:   \eta<x_0(t)<\nu, \mbox{  } 1-\nu<x_1(t)<1-\eta, \mbox{ }\forall t\in[0,T]\right\}\end{gather*}
and the homotopy $\mathcal{H}:[0,1]\times C_T\times C_T\to C_T\times C_T,$ defined by
\begin{equation*} 
	\mathcal{H}[\lambda,x_0,x_1]=(x_0,x_1)-(\mathcal P[x_0,x_1]+\mathcal Q\mathcal N[x_0,x_1]+\lambda \mathcal K(I-\mathcal Q)\mathcal N[x_0,x_1]).
\end{equation*}
Before demonstrating the admissibility of the homotopy, we first compute $d_{LS}(\mathcal{H}[0,\cdot],\Omega,0)$. This computation is particularly straightforward due to the fact that $(I-\mathcal{H}[0,\cdot])(\overline{\Omega})\subseteq\mathbb{R}^2$. Exploiting this condition allows us to reduce its calculation to the finite-dimensional case, namely the Brouwer degree. To accomplish this, we formally establish the following result:
\begin{Pro}\label{grado_finit}
      $d_{LS}(\mathcal{H}[0,\cdot],\Omega,0)=1$.
 \end{Pro}
 \begin{proof}
     As mentioned before, $(I-\mathcal{H}[0,\cdot])(\overline{\Omega})\subseteq\mathbb{R}^2$ implies that
\begin{equation*}
    d_{LS}(\mathcal{H}[0,\cdot],\Omega,0)=d_{B}(\mathcal{H}[0,\cdot]_{\overline{\Omega\cap\mathbb{R}^2}},\Omega\cap\mathbb{R}^2,0)=d_{B}(- \mathcal{Q}\mathcal{N}[\cdot]_{\overline{\Omega\cap\mathbb{R}^2}},\Omega\cap\mathbb{R}^2,0).
\end{equation*}
If $\mathcal{Q}\mathcal{N}(x,y)=0$, then we obtain the system
\begin{equation}\label{eq:SystemEquilibria}
	\begin{split}
		&\overline{f_0}(1-\mu)x+\overline{f_1} \xi y-x(\overline{f_0}x+\overline{f_1}y)=0,  \\
		&\overline{f_0}\mu x+\overline{f_1}(1-\xi)y-y(\overline{f_0}x+\overline{f_1}y)=0.
	\end{split}
\end{equation}
From system \eqref{eq:SystemEquilibria} we get
\begin{equation}\label{y=}
    \overline{f_1}y(x-\xi)=\overline{f_0}x((1-\mu)-x)
\end{equation}
\begin{equation}\label{x=}
    \overline{f_0}x(y-\mu)= \overline{f_1}y((1-\xi)-y).
\end{equation}
Furthermore  from \eqref{y=} and \eqref{x=} we obtain, respectively,
\begin{equation}\label{y'=}
    \dfrac{dy}{dx}=-\left(\frac{\overline{f_0}}{\overline{f_1}}\right)\frac{x(x-\xi) +\xi((1-\mu)-x)}{(x-\xi)^2},\quad
    \dfrac{dx}{dy}=-\left(\frac{\overline{f_1}}{\overline{f_0}}\right)\frac{y(y-\mu)+\mu((1-\xi)-y)}{(y-\mu)^2}.
\end{equation}
We will prove that $\mathcal{Q}\mathcal{N}(x,y)\neq 0$, $\forall (x,y)\in\partial\left([\eta,\nu]\times[1-\nu,1-\eta]\right)$. For this, we assume that there exists $(x_0,y_0)\in\left(\partial[\eta,\nu]\times[1-\nu,1-\eta]\right)$ that satisfies system \eqref{eq:SystemEquilibria}, and we will rule out this possibility through the following cases.
\begin{itemize}
    \item[Case 1.] If $\xi<1-\mu$, $y\in[1-\nu,1-\eta]$ and $x=\eta=\xi$  our $x=\nu=1-\mu$ then from \eqref{y=} we reach a contradiction.\\      
\item[Case 2.] If $\xi<1-\mu$, $x\in[\eta,\nu]$ and  $y=(1-\nu)=\mu$  our $y=(1-\eta)=1-\xi$  then from \eqref{x=} we reach a contradiction. \\  
    \item[Case 3.]  If $\xi>1-\mu$,   $y\in[1-\nu,1-\eta]$ and $x=\eta$ our $x=\nu$. This case is similar to case 1.
    \item[Case 4.]  If $\xi>1-\mu$,   $x\in[\eta,\nu]$ and $y=1-\nu$ our $x=1-\eta$. This case is similar to case 2.
\end{itemize}
Therefore, we can now compute the Brouwer degree. For this, we assert that the system of equations \eqref{eq:SystemEquilibria} has only one solution. In fact, from Cases 1, 2, 3, and 4 it follows that if $(x_0, y_0)$ satisfies \eqref{eq:SystemEquilibria}, then $(x_0, y_0) \in ]\eta, \nu[ \times ]1 - \nu, 1 - \eta[$. Furthermore, in the case $\xi < 1 - \mu$, we observe the following: considering $(a,y_a) $ a point of \eqref{y=},  $(b,y_b) $ a point of \eqref{x=}, $w_0=\overline{f_0}/\overline{f_1}$   
and   the monotony given by\eqref{y'=} then if   $a,b\to \xi^{+}$  we obtain that
\begin{align*}
    y_b&\to \left(\frac{1}{2} \left(1-\xi (w_0+1)+\sqrt{4 \xi \mu w_0+(\xi (w_0+1)-1)^2} \right)\right)^-\\&<\frac{1}{2} \left(1-\xi (w_0+1)+\sqrt{(\xi (w_0-1)+1)^2}\right)\\
    &=1-\xi\\
    &<y_a\\
    &\to+\infty.
\end{align*}
  and if we consider that $a,b\to(1-\mu)^-$
  we obtain that
\begin{align*}
    y_b& \to \frac{1}{2} \left(\sqrt{((1-\xi)-(1-\mu)w_0)^2+4\mu (1-\mu)  w_0}+(1-\xi)-(1-\mu)w_0  \right)^+ \\
    & > \frac{1}{2} \left(\sqrt{(\mu-(1-\mu)w_0)^2+4 \mu(1-\mu) w_0}+(1-\xi)-(1-\mu)w_0  \right) \\
    & = \frac{1}{2} \left(\mu+(1-\xi) \right) \\
    &>\mu\\
    &>y_a\\
    &\to 0^+.
\end{align*}

 Then from  \eqref{y'=} we can conclude that there is a single point of intersection. A symmetric analysis is followed in the case $\xi > 1 - \mu$, leading to the same conclusion. Thus, we have the following:
\begin{align*}
    &d_{B}(- \mathcal Q\mathcal N[\cdot]_{\overline{\Omega\cap\R^2}},\Omega\cap\R^2,0)\\
    &\qquad\quad=\sgn\left(\ \overline{f_0}(1-\mu-2x_0)-\overline{f_1}y_0)(\overline{f_1}(1-\xi-2y_0)-\overline{f_0}x_0)\overline{f_0}~\overline{f_1}(\xi-x_0)(\mu-y_0) \right),
\end{align*}
and considering that $(x_0,y_0)$ satisfy \eqref{y=} and \eqref{x=} we have that
\begin{align*}
    &d_{B}(- \mathcal Q\mathcal N[\cdot]_{\overline{\Omega\cap\R^2}},\Omega\cap\R^2,0)\\
    &\qquad=\sgn\left( \overline{f_0}~\overline{f_1}\left(\frac{(1-\mu)-x_0}{x_0-\xi}+x_0\right)\left(\frac{(1-\xi)-y_0}{y_0-\mu}+y_0\right) -\overline{f_0}~\overline{f_1}(\xi-x_0)(\mu-y_0) \right)\\
    &\qquad= 1.
\end{align*}
  \end{proof}

The upcoming lemma establishes the admissibility of the homotopy $\mathcal H$ and allow us to compute the degree of the operator $I-\mathcal{A}.$

\begin{Lem}\label{cots1}
	Considering all the hypotheses of Theorem \ref{TeoP}. Then, every solution of $$\mathcal H\left[\lambda, x_0, x_1\right] = 0$$
 is contained in $\Omega,$ for all $\lambda\in[0, 1]$ and $t\in\mathbb R.$ 
\end{Lem}
\begin{proof}
	Let $\Lambda$ denote the set of all solutions to $\mathcal H[\lambda, x_0, x_1] = 0$. Assuming by contradiction there exists $\lambda \in [0,1]$ and $(x_0, x_1) \in \Lambda$ such that $(x_0, x_1) \in \partial\Omega$. The case $\lambda=0$ is discarded from the consideration of Proposition \ref{grado_finit}. If $\lambda\in]0,1]$ from the definition of $\mathcal{H}$  we obtain 

\begin{equation*}
	\begin{split}
		\dot x_{0}&=\lambda\left(f_{0}(t)(1-\mu)x_{0}(t-r_0)+f_1(t)\xi x_1(t-r_1)-x_{0}(t)\left(f_0(t)x_0(t-r_0)+f_1(t)x_1(t-r_1)\right)\right),\\
		\dot x_{1}&=\lambda\left(f_{0}(t)\mu x_{0}(t-r_0)+f_{1}(t)(1-\xi)x_{1}(t-r_1)-x_{1}(t)\left(f_0(t)x_0(t-r_0)+f_1(t)x_1(t-r_1)\right)\right).
	\end{split}
\end{equation*}

	  Let $ t_{x}, t_{y}\in [0,T]$ such that  $\dot x_0(t_{x})=0$ and $\dot x_1(t_{y})=0$,
	\begin{eqnarray}
		\label{x'}   x_1(t_x-r_{1}(t_x))(\xi-x_0(t_x))&=&w(t_x) x_0(t_x-r_{0}(t_x))(x_0(t_x)-1+\mu), \\
		\label{y'} x_1(t_y-r_{1}(t_y))(1-\xi-x_1(t_y))&=&w(t_y)x_0(t_y-r_{0}(t_y))(x_1(t_y)-\mu),
	\end{eqnarray}
	where $w=f_0/f_1.$ We have the following cases: If $x_0(t_x)=\max_{t\in[0,T]}x_0(t)=\nu=\xi$ (similar in the case $\nu=1-\mu$) or $x(t_x)=\min_{t\in[0,T]}x_0(t)=\eta=1-\mu$ (similar in the case $\eta=\xi$), then from \eqref{x'}
	\begin{align*}
		0= x_1(t_x-r_{1}(t_x))(\xi-\nu)&=w(t_x)x_0(t_x-r_{0}(t_x))(\nu-1+\mu)>0,\\
		0< x_1(t_x-r_{1}(t_x))(\xi-\eta)&=w(t_x)x_0(t_x-r_{0}(t_x))(\eta-1+\mu)=0,
	\end{align*}
	respectively, which implies a contradiction in both cases. On the other hand,  if  $x_1(t_y)=\max_{t\in[0,T]}x_1(t)=1-\eta=\mu$ (similar in the case $\eta=\xi$) or $x_1(t_y)=\min_{t\in[0,T]}x_1(t)=1-\nu=1-\xi$ (similar in the case $\nu=1-\mu$), then from \eqref{y'} 
	\begin{align*}
		0> x_1(t_y-r_{1}(t_y))(1-\xi-(1-\eta))&=w(t_y)x_0(t_y-r_{0}(t_y))((1-\eta)-\mu)=0,\\
		0= x_1(t_y-r_{1}(t_y))(1-\xi-(1-\nu))&=w(t_y)x_0(t_y-r_{0}(t_y))((1-\nu)-\mu)<0,
	\end{align*}
 	respectively, which implies a contradiction in both cases.
\end{proof}

\begin{proof}[Proof of Theorem \ref{TeoP}]
	By applying Lemma \ref{cots1}, it follows that $\mathcal{H}$ is an admissible homotopy. Consequently,
	\begin{equation*} 
		\begin{split}
			d_{LS}(I-\mathcal{A},\Omega,0)&=d_{LS}(\mathcal{H}[1,\cdot],\Omega,0)\\
			&=d_{LS}(\mathcal{H}[0,\cdot],\Omega,0)\\
   &=1.
		\end{split}
	\end{equation*} 
	Therefore \eqref{eq} has at least one solution in $\Omega$.  Additionally, considering Proposition \ref{Perm_esp_inv}, the proof of Theorem \ref{TeoP} is concluded.
\end{proof}

\begin{proof}[Proof of Theorem \ref{Teo2}] 
	Assuming the fulfillment of all the hypotheses in Theorem \ref{Teo2}. The initial part of the proof follows from Proposition \ref{Perm_esp_inv}. To establish the remaining aspects, we will proceed by considering different cases:
	\begin{itemize}
    \item[Case 1.]  Assume $1-\mu>\xi$ and  $x_0'(t_x)=0$. If $x(t_x)=M_{x_0}=\max_{t\in[0,T]}\{x_0(t)\}\geq 1-\mu$ then from equation \eqref{x'} we get
	\[0> x_1\left(t_x-r_{1}(t_x)\right)(\xi-M_{x_0})=w(t_x)x_0(t_x-r_{0}(t_x))(M_{x_0}-1+\mu)\geq 0.\]
	This is a contradiction. Similarly if $x(t_x)=m_{x_0}=\min_{t\in[0,T]}\{x_0(t)\}\leq \xi$ then from equation \eqref{x'} we have
	\[0\leq x_1(t_x-r_{1}(t_x))(\xi-m_{x_0})=w(t_x)x_0(t_x-r_{0}(t_x))(m_{x_0}-1+\mu)< 0,\] which is also a contradiction. 
 
 On the other hand, since $1-\mu>\xi$ then $1-\xi>\mu$. If we assume that $x_1(t_y)=m_{x_1}\leq \mu$ then from equation \eqref{y'} we obtain
	\[0<x_1(t_y-r_{1}(t_y))(1-\xi-m_{x_1} )=w(t_y)x_0(t_y-r_{0}(t_y))(m_{x_1}-\mu)\leq0,\]
	leading to a contradiction. Similarly if $x_1(t_y)=M_{x_1}\geq 1-\xi$ then from equation \eqref{y'} we obtain
	\[0\geq x_1(t_y-r_{1}(t_y))(1-\xi-M_{x_1} )=w(t_y)x_0(t_y-r_{0}(t_y))(M_{x_1}-\mu)>0.\]
	This is a contradiction.
 
	\item[Case 2.]  Assume $1-\mu<\xi.$ The proof is analogous to the previous case.
	
	\item[Case 3.]  Assume $1-\mu=\xi$ and suppose that $M_{x_0}> \xi$, then from equation \eqref{x'} we have that
	$$0> x_1(t_x-r_{1}(t_x))(\xi-M_{x_0})=w(t_x)x_0(t_x-r_{0}(t_x))(M_{x_0}-1+\mu)> 0.$$
 This is a contradiction 	then  $M_{x_0}\leq \xi$. Similarly if $m_{x_0}< \xi$, then from equation \eqref{x'} we have that
	$$0< x_1(t_x-r_{1}(t_x))(\xi-m_{x_0})=w(t_x)x_0(t_x-r_{0}(t_x))(m_{x_0}-1+\mu)< 0.$$ This is a contradiction then $\xi\leq m_{x_0}$. Therefore $x_0(t)=\xi,$ for all $ t\in[0,T]$. Similarly, it is shown that $x_1(t)=1-\xi,$ for all $ t\in[0,T]$.
 \end{itemize}
\end{proof}

\section{Conclusions}
Quasispecies theory has fundamentally transformed our understanding of the dynamics of replicators such as macromolecules or cells subject to large mutation rates~\cite{Eigen1971,Eigen1988}. Initially applied to prebiotic evolution, the theory was later extended to RNA viruses~\cite{Revull2021,Mas2004,Perales2020} and genetically unstable cancer cells~\cite{Sole2003,Sole2006,Brumera2006}. Early investigations highlighted the heterogeneous nature of quasispecies, where a master sequence is surrounded by a cloud of mutants that stabilize at mutation-selection balance~\cite{Eigen1971,Eigen1988}. This insight has significantly impacted virology, revealing the intricate composition and dynamics of viral populations during infections, with initial experimental demonstrations involving RNA bacteriophage Q$\beta$~\cite{Domingo1978}. Subsequently, the population heterogeneity conceptualized by quasispecies theory has been validated in various viruses, including foot-and-mouth disease virus~\cite{Domingo1980,Sobrino1983}, vesicular stomatitis virus~\cite{Holland1979,Holland1982}, hepatitis viruses~\cite{Martell1992,Davis1999,Mas2004,Perales2020}, and SARS-CoV-2~\cite{Domingo2023}.

A key prediction of quasispecies theory is the error threshold or error catastrophe~\cite{Eigen1971,Domingo1988,Biebricher2005}, where excessive mutation rates lead to the loss of the master sequence, leaving only mutant sequences~\cite{Eigen1971,Eigen1988,Bull2008}. This concept underlies lethal mutagenesis, a strategy to eradicate viral populations by inducing high mutation rates through mutagenic agents~\cite{Bull2008}. Evidence for lethal mutagenesis has been observed in several viruses, including HIV-1~\cite{Loeb1999,Dapp2013}, poliovirus~\cite{Crotty2001}, and hepatitis C virus~\cite{Prieto2013,Avila2016}, among others~\cite{Perales2020}. Moreover, lethal defection, involving the extinction of viral populations due to defective viral genomes, has been identified in lymphocytic choriomeningitis virus~\cite{Grande2005}. Although initial models of quasispecies theory assumed deterministic dynamics and simple fitness landscapes~\cite{Eigen1971,Eigen1988,Swetina1982}, recent efforts have incorporated more realistic scenarios such as finite populations~\cite{Nowak1989,Sardanyes2008,Sardanyes2009}, stochastic effects~\cite{Sole2004,Sardanyes2011,Ari2016}, more complex fitness landscapes~\cite{Wilke2001a,Wilke2001c,Sardanyes2009,Elena2010,Josep2014}, or viral complementation~\cite{Sardanyes2010}. These advancements continue to refine our understanding of RNA virus dynamics, although some aspects, like the impact of time lags in RNA synthesis and periodic replicative fitness, remain largely underexplored.

In this article we investigate the quasispecies model under the single-peak fitness landscape framework. This fitness landscape is one of the simplest ones and allow for analytical exploration. Despite this simplicity, it has been used to characterize quasispecies complexity features in hepatitis C-infected patients~\cite{Mas2004,Sole2006}. We have analyzed different scenarios for this model considering both no backward and backward mutations. Specifically, we demonstrated that periodic orbits exist when backward mutation and periodic fluctuations are present, regardless of time lags. Without backward mutation, neither periodic fluctuations nor time lags result in periodic orbits. In scenarios with periodic fluctuations, solutions converge exponentially to a periodic oscillation around the equilibria associated with a constant replication rate. Concerning the error catastrophe, this hypothesis holds true without backward mutation. However, backward mutations under the given fitness landscape mitigate the error catastrophe bifurcation. 

\section*{Acknowledgment}
The authors ET and FC thank Adri\'an G\'omez for fruitful discussion and guidance on time-delayed systems. Support from Research Agencies of Chile and Universidad del Bío-Bío is acknowledged, they came in the form of research projects GI2310532-VRIP-UBB and ANID PhD/2021-21210522. ET thanks the Centre de Recerca Matem\`atica and especially J. Sardany\'es  for their hospitality and continuous support during the research stay in the first semester of 2024, which allowed us to carry out this research. JS wants to thank Esteban Domingo, Santiago Elena, Ricard Sol\'e, and Celia Perales for insightful discussions on quasispecies theory and viral quasispecies. JS has been funded by a Ramón y Cajal Fellowship (RYC-2017-22243) and by  grant PID2021-127896OB-I00 funded by MCIN/AEI/10.13039/501100011033 ”ERDF A way of making Europe”. We also thank Generalitat de Catalunya CERCA Program for institutional support and the support from María de Maeztu Program for Units of Excellence in R$\&$D grant CEX2020-001084-M (JS). 

\bibliographystyle{abbrv}
\bibliography{CuasiDelay}
\end{document}